\documentclass[11pt]{amsart}
\usepackage{graphicx, amssymb}

\newtheorem{thm}{Theorem}[section]
\newtheorem{cor}[thm]{Corollary}
\newtheorem{lem}[thm]{Lemma}
\newtheorem{prop}[thm]{Proposition}
\theoremstyle{definition}
\newtheorem{defn}[thm]{Definition}

\theoremstyle{remark}
\newtheorem{rem}[thm]{Remark}
\newtheorem{exa}[thm]{Example}
\numberwithin{equation}{section}

\newcommand{\set}[1]{\left\{#1\right\}}
\newcommand{\Real}{\mathbb R}
\newcommand{\Natural}{\mathbb N}
\newcommand{\Natura}{\mathbb{N} \cup \{ \infty \} }

\newcommand{\B}{\mathcal{B}}
\newcommand{\such}{\, | \,}
\newcommand{\limn}{\lim_{n \to \infty}}
\newcommand{\Pre}{\mathcal{P}}
\newcommand{\prob}{\mathbb{P}}
\newcommand{\oprob}{\overline{\mathbb{P}}}
\newcommand{\plim}{\oprob \text{-} \lim}
\newcommand{\plimn}{\plim_{n \to \infty}}
\newcommand{\climn}{\mathsf{C} \text{-} \lim_{n \to \infty}}
\newcommand{\ucp}{\mathsf{uc} \oprob \text{-} \lim}
\newcommand{\ucpn}{\ucp_{n \to \infty}}
\newcommand{\ucpeps}{\ucp_{\epsilon \downarrow 0}}

\newcommand{\oexpec}{\overline{\mathbb{E}}}

\newcommand{\Lb}{\mathbf{L}}
\newcommand{\F}{\mathcal{F}}
\newcommand{\G}{\mathcal{G}}
\newcommand{\cadlag}{c\`adl\`ag}

\newcommand{\ud}{\, \mathrm d}
\newcommand{\inner}[2]{\left \langle #1 , \, #2 \right \rangle}
\newcommand{\innerw}[2]{\left \langle #1 , #2 \right \rangle}
\newcommand{\num}{num\'eraire}
\newcommand{\X}{\mathcal{X}}

\newcommand{\pr}{\mathsf{proj}}

\newcommand{\oS}{\overline{\mathcal{S}}}
\newcommand{\oSn}{\overline{\mathcal{S}} \text{-} \limn}
\newcommand{\oSeps}{\overline{\mathcal{S}} \text{-} \lim_{\epsilon \downarrow 0}}

\newcommand{\ta}{\widetilde{a}}
\newcommand{\oa}{\overline{a}}
\newcommand{\tphi}{\widetilde{\varphi}}
\newcommand{\ophi}{\overline{\varphi}}

\newcommand{\og}{\overline{\mathfrak{g}}}
\newcommand{\tg}{\widetilde{\mathfrak{g}}}

\newcommand{\dist}{\mathsf{dist}}

\newcommand{\tX}{\widetilde{X}}
\newcommand{\hX}{\widehat{X}}

\newcommand{\pare}[1]{\left(#1\right)}
\newcommand{\bra}[1]{\left[#1\right]}
\newcommand{\dbra}[1]{[\kern-0.15em[ #1 ]\kern-0.15em]}
\newcommand{\dbraco}[1]{[\kern-0.15em[ #1 [\kern-0.15em[}

\newcommand{\K}{\mathfrak{K}}
\newcommand{\Nu}{\mathfrak{N}}

\newcommand{\bF}{\mathbf{F}}
\newcommand{\ubF}{\underline{\mathbf{F}}}
\newcommand{\obF}{\overline{\mathbf{F}}}

\newcommand{\oF}{\overline{\mathcal{F}}}

\newcommand{\data}{{(\bF, \, \prob, \, \K)}}
\newcommand{\datan}{{(\bF^n, \, \prob^n, \, \K^n)}}

\newcommand{\indic}{\mathbb{I}}

\newcommand{\kt}{\mathfrak{t}}

\newcommand{\absco}{{<\kern-0.53em<}}

\begin{document}

\title[The continuous behavior of the num\'eraire portfolio]{The continuous behavior of the num\'eraire portfolio under small changes in information structure, probabilistic views and investment constraints}%
\author{Constantinos Kardaras}%
\address{Constantinos Kardaras, Mathematics and Statistics Department, Boston University, 111 Cummington Street, Boston, MA 02215, USA.}%
\email{kardaras@bu.edu}%


\date{\today}%
\begin{abstract}
The \textsl{\num \ portfolio} in a financial market  is the unique positive wealth process that makes all other nonnegative wealth processes, when deflated by it, supermartingales. The \num \ portfolio depends on market characteristics, which include: \textbf{(a)} the information flow available to acting agents, given by a \textsl{filtration}; \textbf{(b)} the statistical evolution of the asset prices and, more generally, the states of nature, given by a \textsl{probability} measure; and \textbf{(c)} possible restrictions that acting agents might be facing on available investment strategies, modeled by a \textsl{constraints} set. In a financial market with continuous-path asset prices, we establish the stable behavior of the \num \ portfolio when each of the aforementioned market parameters is changed in an infinitesimal way.
\end{abstract}

\maketitle

\setcounter{section}{-1}

\section{Introduction}

Within the class of expected utility maximization problems in the theory of \emph{Financial Economics}, the one involving expected \emph{logarithmic} utility plays a central role. Its importance can be understood by going as back as \cite{KELLY}, where the optimal exponential growth for a gambler's wealth was discovered from an information-theoretic point of view. In general semimartingale models, it is the \emph{only} case of utility where an explicit solution can be given in terms of the triplet of predictable characteristics, as was carried out in \cite{MR1970286}.

The log-optimal portfolio, when it exists, is the \textsl{\num \ portfolio} (an appellation that was introduced in \cite{Long90}) according to the definition in \cite{MR1849424}: all other wealth processes, when discounted by the log-optimal one, become supermartingales under the historical (statistical, real world) probability. In fact, the \num \ portfolio can exist even in cases where the log-optimal problem does not have a unique solution, which happens when the value of the log-optimal problem is infinite.

The \num \ portfolio depends on the stochastic nature of the financial market. As the output of an optimization problem, it is of importance to ensure that it is stable under small changes in the market parameters. Here, focus is given on the following characteristics:
\begin{itemize}
  \item Available \emph{information} that economic agents have access to.
  \item \emph{Statistical} (or even \emph{subjective}) \emph{views} on the possible future outcomes.
  \item \emph{Investment constraints} (usually, institutionally-enforced) that agents face.
\end{itemize}
Institutionally-enforced constraints can involve, for example, prevention of short sales. Another important restriction that agents face is that of a finite credit limit; their wealth has to remain positive in order to avoid bankruptcy.

\smallskip

The purpose of this work is to guarantee that small deviations from the above market characteristics do not lead to radical changes in the structure of the \num \ portfolio. Naturally, part of the problem is to rigorously define what is meant by ``deviations'' of the market characteristics. In turn, this means that in order to achieve the desired continuous behavior of the \num \ portfolio, certain economically-reasonable topological structures have to be placed on  filtrations, probabilities and constraint sets.

\smallskip

Stability of the \num \ portfolio is a \emph{qualitative} study; there are, however, good \emph{quantitative} reasons to undertake such study. Lately, there has been significant interest in quantifying the \emph{value of insider information}, as measured via the increase in the log-utility of an insider with respect to a non-informed trader. One can check, for example, \cite{MR2223957} and the wealth of references therein. It then becomes plausible to examine \emph{marginal} values of insider information, or of investment freedom. The last question is intimately related to \emph{differentiability} of the \num \ portfolio (or at least of the value of the log-utility maximization problem) with respect to market parameters. Such differentiability would give a first-order approximation of the behavior of the \num \ portfolio. Before seeking conditions ensuring differentiability, which is a possible topic for further research, a zeroth-order study concerning continuity has to be carried out. In the present work, we only scratch the very surface of the problem of differentiability of the \num \ portfolio and the calculation of its derivative.

\smallskip

The structure of the paper as follows. Section \ref{sec: main} sets up the model with continuous asset-price processes, where markets are parameterized via a triplet of data, including information flows, statistical structure and investment constraints that agents face. A ``proximity'' concept for the market parameters is introduced by defining modes of convergence for the three data inputs. Theorem \ref{thm: main} is the result which establishes continuity of the \num \ portfolio in a rather strong sense under convergence of market's data. Then, Section \ref{sec: proof} is dedicated to proving Theorem \ref{thm: main}.

\smallskip

The workable expression that is obtained for the \num \ portfolios allows for a bare-hands approach to proving Theorem \ref{thm: main}. This should be contrasted with the treatment in \cite{LarZit06} and \cite{KarZit07}, where passage to the dual problem, as described in \cite{MR1722287}, is necessary. There, unnatural (from an economical point of view) uniform integrability conditions have to be assumed involving the class of equivalent martingale measures of the market.

The assumption of continuity of the asset-price processes is made for simplifying the presentation. (It should be noted however that an elementary example in \S \ref{subsec: jump counterexa} shows that the result of Theorem \ref{thm: main} is not valid without additional control on the agent's constraints.) Even by assuming continuity of the asset-price processes, one cannot completely avoid dealing with jumps in the proof of Theorem \ref{thm: main}. Changing from the probability measure of one market to the one of another, as has to be done, results in the appearance of martingale density processes with possible jump components. These technical complications make the proof of Theorem \ref{thm: main} somewhat lengthy.

\section{The Result on the Continuous Behavior of the Num\'eraire Portfolio} \label{sec: main}

\subsection{The set-up}

Every stochastic process in the sequel is defined on a \textsl{stochastic basis} $(\Omega, \, \oF, \, \obF, \, \oprob)$. Here, $\oprob$ is a probability on $(\Omega, \oF)$, where $\oF$ is a $\sigma$-algebra that will make all involved random variables measurable. Further, $\obF = \pare{\oF_t}_{t \in \Real_+}$ is a ``large'' filtration  that will dominate all other filtrations that will appear. Of course, $\oF_t \subseteq \oF$ for all $t \in \Real_+$ and $\obF$ is assumed to satisfy the \textsl{usual hypotheses} of right-continuity and saturation by $\oprob$-null sets.

\subsubsection{Assets and investing}

The price-processes of $d$ traded financial assets, where $d \in \Natural = \set{1, 2, \ldots}$, are denoted by $S^1, \ldots, S^d$. All processes $S^i$, $i=1, \ldots, d$, are $\obF$-adapted and are assumed to have been discounted by a ``baseline'' asset that will act as a deflator for the denomination of all wealth processes.

The minimal filtration that  makes $S$ adapted and satisfies the usual hypotheses will be denoted by $\ubF$. Since $S$ is $\obF$-adapted, $\ubF \, \subseteq \, \obF$. In the sequel, the information flow of economic agents acting in the market will be modeled via elements $\bF$ such that
\begin{equation} \label{SAND} \tag{INFO}
\bF \text{ is a filtration satisfying the usual hypotheses, and } \ubF \, \subseteq \, \bF \subseteq \obF.
\end{equation}
We shall also model statistical, or subjective, views of economic agents via $\prob$, where
\begin{equation} \label{LOC-EQUIV} \tag{P-LOC-EQUIV}
\prob \text{ is a probability, with }\prob \sim \oprob \text{  on } \oF_T \text{ holding for all } T \in \Real_+.
\end{equation}

The following innocuous assumption on the structure of the $S$ will be in force throughout:
\begin{equation} \label{SEMI-MART} \tag{CON-SEMI-MART}
S \text{ is a } (\obF, \oprob) \text{-semimartingale with } \oprob \text{-a.s. continuous paths}.
\end{equation}
For a pair $(\bF, \prob)$ satisfying \eqref{SAND} and \eqref{LOC-EQUIV}, \eqref{SEMI-MART} implies that $S$ is a $(\bF, \prob)$-semimartingale. Therefore, one can define the class of all possible \emph{nonnegative} wealth processes starting from (normalized) unit initial capital for a market in which the information-probability structure is given by $(\bF, \prob)$:
\begin{equation} \label{eq: wealth processes, all}
\X^{\bF} \, := \, \set{ X^{\vartheta} \equiv 1 + \int_0^\cdot \vartheta_t \ud S_t \ \Big| \ \vartheta \text{ is } \bF\text{-predictable and } S\text{-integrable, and } X^{\vartheta} \geq 0, \ \prob \text{-a.s.}}
\end{equation}
The dependence on $\prob$ from $\X^\bF$ in \eqref{eq: wealth processes, all} above is suppressed, simply because there is no dependence in view of \eqref{LOC-EQUIV}. The following structural assumption on the class of wealth processes will be in force throughout.
\begin{equation} \label{NUPBR} \tag{NUPBR}
\downarrow \lim_{m \to \infty} \sup_{X \in \X^{\obF}} \oprob \, [X_T > m] = 0, \text{ for all } T \in \Real_+.
\end{equation}
(Note that ``$\downarrow \lim$'' denotes a \emph{nonincreasing} limit.) In other words, the set $\{ X_T \such X \in \X^{\obF} \}$ is bounded in $\oprob$-probability for all $T \in \Real_+$. For a pair $(\bF, \prob)$ satisfying \eqref{SAND} and \eqref{LOC-EQUIV}, \eqref{NUPBR} implies that $\{X_T \such X \in \X^{\bF} \}$ is bounded in $\prob$-probability for all $T \in \Real_+$.

\begin{rem}

According to \cite{MR2335830}, condition \eqref{NUPBR}, an acronym for \textsl{No Unbounded Profit with Bounded Risk}, is equivalent to existence of the \num \ portfolio (see \S\ref{subsubsec: num portf} below) for any pair $(\bF, \, \prob)$ that satisfies \eqref{SAND} and \eqref{LOC-EQUIV}. Since this work is aimed at studying stability of the \num \ portfolio, \eqref{NUPBR} is a minimal structural assumption.
\end{rem}

\subsubsection{Constraints on investment}

Fix some pair $(\bF, \, \prob)$ corresponding to the information-probability structure of the financial market. Agents in this market might be facing constraints on possible investment strategies, which we now formally describe. Consider a set-valued process $\K : \Omega \times \Real_+ \mapsto \B(\Real^d)$, where $\B(\Real^d)$ denotes the class of Borel subsets of $\Real^d$. A process in $X^{\vartheta} \in \X^{\bF}$ will be called \textsl{$\K$-constrained} if $\vartheta_t(\omega) \in X_t^{\vartheta}(\omega) \K(\omega, t)$ for all $(\omega, t) \in \Omega \times \Real_+$; in short, $\vartheta \in X^{\vartheta} \K$. (Investment constraints of this kind, but where no dependence of the constraint sets on $(\omega, t) \in \Omega \times \Real_+$ is involved, appear in an It\^o-process modeling context in the literature in \cite{MR1189418}.) We denote by $\X^{(\bF, \K)}$ the class of all $\K$-constrained wealth processes in $\X^{\bF}$; namely, $\X^{(\bF, \K)}  \, := \, \set{ X^{\vartheta} \in \X^{\bF} \ \big| \ \vartheta \in X^{\vartheta} \K}$. For $X^{\vartheta} \in \X^{(\bF, \K)}$, both $\vartheta$ and $X^{\vartheta}$ are $\bF$-predictable. It makes sense, both from a mathematical and a financial point of view, to give the constraints set a predictable structure as well. A set-valued process $\K$ will be called \textsl{$\bF$-predictable} if $\{ (\omega, t) \such \K(\omega, t) \cap K \neq \emptyset \}$ is an $\bF$-predictable set for all \emph{compact} $K \subseteq \Real^d$. For more information on this kind of measurability, see Appendix 1 of \cite{MR2335830}, or Chapter 17 of \cite{MR1717083} for a more general treatment. Further, it is financially reasonable to put some closedness and convexity structure on $\K$. We call $\K$ \textsl{closed} and \textsl{convex} if $\K(\omega, t)$ has these properties for all $(\omega, t) \in \Omega \times \Real_+$.

\subsubsection{Financial market data}

Before the formal definition of the financial market's \emph{data} is given, we tackle degeneracies that might appear in the asset-price process. Call $G := \mathsf{trace} [S, S]$, where ``$\mathsf{trace}$'' denotes the trace operator on matrices and $[S, S]$ denotes the continuous, $(d \times d)$-matrix-valued quadratic covariation process of $S$. It is straightforward that $G$ is $\ubF$-predictable and nondecreasing. There exists a $(d \times d)$-nonnegative-definite-matrix-valued, $\ubF$-predictable process $c$ such that $[S,S] = \int_0^\cdot c_t \, \ud G_t$ (in obvious matrix notation). Define $\Nu := \set{ x \in \Real^d \such c x = 0 }$, where the dependence of $\Nu$ on $(\omega, t)$ is suppressed; $\Nu$ is $\ubF$-predictable and takes values in linear subspaces of $\Real^d$. We denote by $\Nu^\bot$ the orthogonal complement of $\Nu$; this is also a $\ubF$-predictable, $\Real^d$-subset-valued process. (The facts that $\Nu$ and $\Nu^\bot$ are $\ubF$-predictable follow by the results of Appendix 1 of \cite{MR2335830}.) Now, pick any $\bF$-predictable, $S$-integrable and $\Nu$-valued process $\vartheta$. The gains process $\int_0^\cdot \vartheta_t \ud S_t$ has null quadratic variation. Under \eqref{NUPBR}, $\int_0^\cdot \vartheta_t \ud S_t$ is identically equal to zero. Therefore, any agent should be free to invest in these $\Nu$-valued strategies, since they result in zero wealth. In other words, we should have, in compact notation, $\Nu \subseteq \K$.

We are ready to give the modeling structure of the financial market environment.

\begin{defn}
A triplet $\data$ will be called \textsl{financial market data}, if $\bF$ satisfies \eqref{SAND}, $\prob$ satisfies \eqref{LOC-EQUIV}, and $\K$ is an $\bF$-predictable, convex and closed $\Real^d$-set valued process such that $\Nu \subseteq \K$.
\end{defn}

\subsubsection{Num\'eraire portfolios} \label{subsubsec: num portf}

Under \eqref{SEMI-MART}, and for a pair $(\bF, \, \prob)$ that satisfies \eqref{SAND} and \eqref{LOC-EQUIV}, decompose $S = A^{(\bF, \prob)} + M^{(\bF, \prob)}$, where $A^{(\bF, \prob)}$ is an $\bF$-adapted, continuous process of locally finite variation, and $M^{(\bF, \prob)}$ is a $\bF$-local $\prob$-martingale. Assumption \eqref{NUPBR} implies that there exists an $\bF$-predictable process $a^{(\bF, \prob)}$ such that
\begin{equation} \label{NUPBR pred}
A^{(\bF, \prob)} = \int_0^\cdot \pare{ c_t a^{(\bF, \prob)}_t} \ud G_t, \text{ where } \int_0^T \inner{a^{(\bF, \prob)}_t}{c_t a^{(\bF, \prob)}_t} \ud G_t < \infty \text{ for all } T \in \Real_+.
\end{equation}
(This last fact was already present in \cite{MR1384360}, although not stated this way. The previous \textsl{structural conditions} \eqref{NUPBR pred} have also appeard in \cite{MR1183992} and \cite{MR1353193}.)

In the financial market with data $\data$, the \textsl{\num \ portfolio} is the \emph{unique} wealth process $\hX^\data \in \X^{(\bF, \, \K)}$ with the property that $X / \hX^\data$ is $(\bF, \prob)$-supermartingale for all $X \in \X^{(\bF, \, \K)}$. (For a complete list of the properties of the \num \ portfolio in connection to what is described here, one could check \cite{MR2335830}.) It can be shown  that the \num \ portfolio is the one that maximizes the growth of the wealth process, where the $(\bF, \, \prob)$-\textsl{growth} of a wealth process $X \in \X^{(\bF, \, \K)}$ with $\prob[X_t > 0, \text{ for all } t \in \Real_+] = 1$ is defined to be the finite variation part of $\log(X)$ in its $(\bF, \, \prob)$-semimartingale decomposition.

We shall now give a more concrete description of the \num \ portfolio. Start with some $X \in \X^{(\bF, \, \K)}$ such that $\prob[X_t > 0, \text{ for all } t \in \Real_+] = 1$, and consider the $\bF$-predictable, $d$-dimensional process $\pi$ defined implicitly via $\ud X_t = X_t \innerw{\pi_t}{\ud S_t}$. Using It\^o's formula and \eqref{NUPBR pred}, the $(\bF, \, \prob)$-growth of $X$ is easily seen to be equal to $\int_0^\cdot \big( \langle \pi_t, c_t a_t^{(\bF, \, \prob)}  \rangle -  \inner{\pi_t}{c_t \pi_t} / 2 \big) \ud G_t$.
As discussed previously, if $X^\data$ is to be the \num \ portfolio, it must have maximal growth. Therefore, let $\varphi^{(\bF, \, \prob, \, \K)}$ be the \emph{unique} $\bF$-predictable, $d$-dimensional, $(\K \cap \Nu^\bot)$-valued process that satisfies
\begin{equation} \label{eq: phins}
\varphi^\data (\omega, t) \, := \, \arg \max_{f \in \K \cap \Nu^\bot (\omega, t)} \pare{ \inner{f}{ c (\omega, t) a^{(\bF, \, \prob)} (\omega, t)} - \frac{1}{2} \inner{f}{c (\omega, t) f}},
\end{equation}
 for all $(\omega, t) \in \Omega \times \Real_+$. (If $\K = \Real^d$, $\varphi^{(\bF, \, \prob, \, \Real^d)} = a^{(\bF, \, \prob)}$.) The process $\varphi^{(\bF, \, \prob, \, \K)}$ is well-defined; this follows from the fact that the maximization problem \eqref{eq: phins} defining $\varphi^{(\bF, \, \prob, \, \Real^d)}$ is strictly concave and coercive on the closed convex set $\K \cap \Nu^\bot$. Its $\bF$-predictability follows from the corresponding property of the inputs $\K \cap \Nu^\bot$, $c$, $a^{(\bF, \, \prob)}$; again, we send the interested reader to Appendix 1 of \cite{MR2335830}. It then follows that \textsl{the $\data$-\num \ portfolio} $\hX^\data$ satisfies $\hX^\data_0 = 1$ and the dynamics $\ud \hX^\data_t = \hX^\data_t \big \langle \varphi_t^\data,  \ud S_t \big \rangle$ for $t \in \Real_+$. In other words, in logarithmic terms,
\begin{equation} \label{eq: numeraire}
\log \hX^{(\bF, \, \prob, \, \K)} \, := \, - \frac{1}{2} \int_0^\cdot \inner{\varphi^\data_t}{ c_t \varphi^\data_t} \ud G_t + \int_0^\cdot \varphi^\data_t \, \ud S_t.
\end{equation}
Indeed, it is straightforward to check that $\hX^{(\bF, \, \prob, \, \K)}$ as defined above is such that $X / \hX^{(\bF, \, \prob, \, \K)}$ is a $(\bF, \prob)$-supermartingale for all $X \in \X^{(\bF, \, \K)}$.

\subsection{Convergence assumptions}

In order to formulate the question of continuous behavior of the \num \ portfolio, several markets will be considered. For each $n \in \Natura$, the market structure will be modeled via the data $\datan$. The limiting behavior of the data triplets will be given in the paragraphs that follow. What is sought after is convergence, as $n \to \infty$, of the $n^{\text{th}}$ market's \num \ portfolio to the \num \ portfolio of the market corresponding to $n = \infty$.

First, convergence of filtrations is settled. Let us give some intuition. Assume, for simplicity, that all markets work under that same probabilistic structure, given by $\oprob$. For any $A \in \G$, an agent with information $\bF^n$ can only \emph{project} at each time $t \in \Real_+$ the the conditional probability $\oprob [A \such \F^n_t]$ that $A$ will happen or not. A natural way to define convergence of $(\bF^n)_{n \in \Natural}$ then would be to require that $\oprob [A \such \F^n_t]$ converges in $\oprob$-probability to $\oprob [A \such \F_t]$, at least pointwise for all $t \in \Real_+$. We ask something somewhat weaker.
\begin{equation} \label{F-CONV} \tag{F-CONV}
\plimn \int_0^T \left| \oprob[A \such \F^n_t] - \oprob[A \such \F^\infty_t] \right| \ud G_t = 0, \text{ for all } A \in \oF \text{ and } T \in \Real_+.
\end{equation}
Note that \eqref{F-CONV} certainly holds in the case where $(\bF^n)_{n \in \Natural}$ converges monotonically to $\bF^\infty$, in the sense that $\uparrow \limn \F^n_T = \F^\infty_T$ or $\downarrow \limn \F^n_T = \F^\infty_T$ for all $T \in \Real_+$, in view of the martingale convergence theorems.

The assumption on convergence of $(\prob^n)_{n \in \Natural}$ to $\prob^\infty$ is:
\begin{equation} \label{P-CONV} \tag{P-CONV}
\plimn \pare{ \frac{\ud \prob^n}{\ud \oprob} \bigg|_{\oF_T} } = \frac{\ud \prob^\infty}{\ud \oprob}\bigg|_{\oF_T}, \text{ for all } T \in \Real_+.
\end{equation}
Note that, as a consequence of Scheffe's lemma, \eqref{P-CONV} is equivalent to saying that $(\prob^n)_{n \in \Natural}$ converges in total variation to $\prob^\infty$ on $\oF_T$ for all $T \in \Real_+$.

We turn to the constraints sets. For two subsets $K \subseteq \Real^d$ and $K' \subseteq \Real^d$ define their \textsl{Hausdorff distance}
\begin{equation} \label{eq: dist}
\dist (K, K') := \max \left\{ \sup_{x \in K} \inf_{x' \in K'} |x - x'|, \, \sup_{x' \in K'} \inf_{x \in K} |x - x'| \right\}.
\end{equation}
For $m \in \Real_+$, let $B(m) := \{ x \in \Real^d \such |x| \leq m \}$. For a collection $(K^n)_{n \in \Natura}$ of subsets of $\Real^d$, define
\[
\climn K^n = K^\infty \text{ if and only if } \limn \dist \left( K^n \cap B(m), K^\infty \cap B(m) \right) = 0, \text{ for all } m \in \Real_+.
\]
Note that this convergence is \emph{weaker} than requiring $\limn \dist(K^n, K^\infty) = 0$ and that it is equivalent to saying that $K^\infty$ is the \textsl{closed limit} of the sequence $(K^n)_{n \in \Natural}$ (see Definition 3.66, page 109 of \cite{MR1717083}). We then ask that
\begin{equation} \label{C-CONV} \tag{C-CONV}
\climn \K^n(\omega, t) = \K^\infty(\omega, t), \text{ for all } (\omega, t) \in \Omega \times \Real_+.
\end{equation}

\subsection{Stability of the \num \ portfolio}

Continuity of the log-wealth of the \num \ portfolios will be obtained with respect to a strong convergence notion, which is now defined. Consider a collection $(\xi^n)_{n \in \Natura}$, each element being a continuous $(\obF, \, \oprob)$-semimartingale. For $n \in \Natura$, write $\xi^n = \overline{B}^n + \overline{L}^n$, where $\overline{B}^n$ is $\obF$-adapted, continuous and of finite variation and $\overline{L}^n$ is a $\obF$-local $\oprob$-martingale. We say that $(\xi^n)_{n \in \Natural}$ $\oS$-converges to $\xi^\infty$ and write $\oSn \xi^n = \xi^\infty$ if and only if $\plimn \int_0^T |\ud (\overline{B}_t^n - \overline{B}_t^\infty)| = 0$ as well as $\plimn [\overline{L}^n - \overline{L}^\infty, \, \overline{L}^n - \overline{L}^\infty]_T = 0$ hold for all $T \in \Real_+$. By the treatment in \cite{MR568256}, it can be shown that $\oS$-convergence is equivalent to (local, in time) convergence in the \textsl{semimartingale topology} on $(\Omega, \, \oF, \, \obF, \, \oprob)$ that was introduced in \cite{MR544800}. In particular, $\oS$-convergence is stronger than the \textsl{uniform convergence on compacts in probability}: $\oSn \xi^n = \xi^\infty$ implies $\ucpn \xi^n = \xi^\infty$, the last equality meaning $\plimn \sup_{t \in [0, T]} |\xi^n_t - \xi^\infty_t| = 0$, for all $T \in \Real_+$.

\begin{thm} \label{thm: main}
Consider a collection of markets, each with data $\datan$, indexed by $n \in \Natura$. Assume that all \eqref{SEMI-MART}, \eqref{NUPBR}, \eqref{F-CONV}, \eqref{P-CONV} and \eqref{C-CONV} are valid. For $n \in \Natura$, let $\hX^n := \hX^\datan$, be the $n^{\text{th}}$ \num \ portfolio. Then,
\[
\oSn \log \hX^n = \log \hX^\infty.
\]
\end{thm}

The proof of Theorem \ref{thm: main} is given in Section \ref{sec: proof}. It is easy to argue why Theorem \ref{thm: main} is true, and this somewhat sets the plan for the proof. For notational simplicity, let $a^n := a^\datan$ and $\varphi^n := \varphi^\datan$ for all $n \in \Natura$. Under \eqref{P-CONV} and \eqref{F-CONV} one would expect that $(a^n)_{n \in \Natural}$ converges in some sense to $a^\infty$. Then, \eqref{C-CONV} and \eqref{eq: phins} should imply that $(\varphi^n)_{n \in \Natural}$ converges (again, in some sense) to $\varphi^\infty$. After that, \eqref{eq: numeraire} makes it very plausible that $(\log \hX^n)_{n \in \Natural}$ should converge to $\log \hX^\infty$. Of course, one has to give precise meaning to these ``senses'' of convergence of the predictable processes. The details of the proof are technical, but more or less follow the above intuitive steps.

\begin{rem}
The result of Theorem \ref{thm: main}, given all its notation and assumptions, implies
\begin{equation} \label{eq: rel error inf}
\limn \prob^\infty \bra{ \sup_{t \in [0, T]} \left| \frac{\hX_t^n - \hX_t^\infty}{\hX_t^\infty} \right| > \epsilon} = 0, \ \text{ for all } T \in \Real_+ \text{ and } \epsilon > 0,
\end{equation}
as well as
\begin{equation} \label{eq: rel error n}
\limn \prob^n \bra{ \sup_{t \in [0, T]} \left| \frac{\hX_t^\infty - \hX_t^n}{\hX_t^n} \right| > \epsilon} = 0, \ \text{ for all } T \in \Real_+ \text{ and } \epsilon > 0.
\end{equation}
Both of the above limiting relationships are incarnations of the fact that \emph{small deviations from information, probability and investment constraints structures will lead to a small relative change in the \num \ portfolio}. While \eqref{eq: rel error inf} is from the point of view of the limiting market, \eqref{eq: rel error n} takes the viewpoint of the approximating markets.
\end{rem}

\subsection{The case of asset-prices with jumps} \label{subsec: jump counterexa}

Theorem \ref{thm: main} need not hold in the case where jumps are present in the asset-price process. A simple discrete one-time-period counterexample is given below; after that, a discussion follows on what the issue is, along with a possible resolution.

\begin{exa} \label{exa: counterexa with jumps}
Consider a one-time-period discrete stochastic basis $(\Omega, \, \oF, \, \obF, \, \oprob)$, where $\obF = (\oF_0, \, \oF_1)$. Suppose that $(\Omega, \, \oF, \,  \oprob)$ is rich enough to accommodate a sequence $(\varepsilon_n)_{n \in \Natural}$ of independent standard normal random variables, as well as some random variable $\eta$, independent of the previous Gaussian sequence with $\oprob[\eta = 1] = p = 1 - \oprob[\eta = -1]$, where $0 < p < 1$. Define a collection $(\bF^n)_{n \in \Natura}$ of filtrations via $\F^n_1 = \oF_1 := \sigma \pare{ (\varepsilon_j)_{j \in \Natural}, \, \eta }$,
 for all $n \in \Natura$ (the information at the terminal date is the same in all markets), as well as
\[
\F^n_0 \, := \, \sigma (\varepsilon_1, \ldots, \varepsilon_n) \textrm{ for all }
n \in \Natural, \textrm{ and } \F^\infty_0 = \oF_0 = \, := \, \sigma \big(
(\varepsilon_j)_{j \in \Natural} \big) = \bigvee_{n \in \Natural}
\F^n_0.
\]
Of course, $(\bF^n)_{n \in \Natural}$ converges monotonically upwards to $\bF^\infty$.

The financial market has one risky asset: $d=1$. With $\varepsilon := \sum_{j =1}^{\infty} 2^{-j} \varepsilon_j$, set $S_0 = 0$ and $S_1 = \varepsilon \eta$. The classical \textsl{No Arbitrage} condition holds for the market with information $\obF = \bF^\infty$, which is the equivalent of \eqref{NUPBR} for discrete-time models. The probabilistic structure is the same in all the markets, given by $\oprob$. Further, no institutionally-enforced constraints are present for agents acting in the market indexed by any $n \in \Natura$.

For the limiting market with $\bF^\infty$-information, the model is just a (conditional) binomial one, since $\varepsilon$ is $\F^\infty_0$-measurable. We have $X^{1, \vartheta}_1 := 1+ \vartheta_0 (S_1 - S_0) = 1 + \vartheta_0
\varepsilon \eta$. Since $\vartheta_0 \in \F_0^\infty \supseteq \sigma(\varepsilon)$, it is easy to see (optimizing the expected log-utility) that the limiting market's \num \ portfolio is such that $\hX^\infty_1 := 1 + (2 p - 1) \eta$. If $p \neq 1/2$, $\prob[\hX^\infty_1 = 1] = 0$.

Consider now the market with information $\bF^n$ for some $n \in \Natural$. Conditional on $\F_0^n$, $\varepsilon$ is independent of $\eta$ and its law is Gaussian with mean $\sum_{j=1}^n 2^{-j} \varepsilon_j$ and variance $1/2^{n+1}$. Since the conditional law of $S_1 - S_0$ is supported on the whole real line, we get $\X^{\bF^n}(1) = \{ 1 \}$. Therefore, for each $n \in \Natural$, every approximating market's \num \ portfolio satisfies $\hX^n_1 := 1$. This obviously does not converge to $\hX^\infty_1$, if $p \neq 1 / 2$.
\end{exa}

In the previous example, all the \emph{nonnegative} wealth process sets $\X^{\bF^n}(1)$ are trivial, but the limiting $\X^{\bF^\infty}(1)$ is non-trivial. Even though there are no institutionally-enforced constraints in the markets, agents \emph{still} have to face the \textsl{natural constraints} $\K_+^n$, $n \in \Natural$, that ensure the positivity of the wealth process. As it turns out, $\K_+^n = \{ 0 \}$ for all $n \in \Natural$, while $\K_+^\infty = \bra{-1 / |\varepsilon|, \, 1 / |\varepsilon|}$. Such behavior is of course absent in the case of continuous-path price processes.

A possible resolution to the previous problem could be the following. In a general discrete-time model, if $\underline{\K}_+$ denotes the natural positivity constraints of the market with information $\ubF$, then $\underline{\K}_+ (\omega, t) \subseteq \K_+^n (\omega, t)$ holds for all $(\omega, t) \in \Omega \times \Real_+$ and in $n \in \Natura$ in view of \eqref{SAND}. If one \emph{forces} from the beginning the additional assumption $\K^n (\omega, t) \subseteq \underline{\K}_+ (\omega, t)$ for all $n \in \Natura$, the problem encountered at Example \ref{exa: counterexa with jumps} ceases to exist, and one should be able to proceed as before.

\subsection{First-order analysis}

Once continuity of the \num \ portfolio is established, the next natural step is to study the \emph{direction} of change given specific changes of the inputs. We provide here a first insight on how the \num \ portfolio changes when we alter \emph{only} the probabilistic structure of the problem, keeping the information fixed and working on the non-constrained case. In more general situations the problem is expected to be rather involved.

For the purposes of this subsection, we shall change the notation slightly. We simply use $\bF = (\F_t)_{t \in \Real_+}$, instead of $\obF$, to denote the common filtration of all agents. Let $\prob^0$ be the ``limiting'' probability (the one that we previously denoted by $\prob^\infty$). Furthermore, let $\prob^1$ be some probability that is equivalent to $\prob^0$, and let $\prob^\epsilon := (1 - \epsilon) \prob^0 + \epsilon \prob^1$. Write $S = A^\epsilon + M^\epsilon$ for the Doob-Meyer decomposition of $S$ under $(\bF, \prob^\epsilon)$; here, $A^\epsilon = \int_0^\cdot c_t a^\epsilon_t \ud G_t$ and $[M^\epsilon, M^\epsilon] = [M^0, M^0] = \int_0^\cdot c_t \ud G_t$ for all $\epsilon \in [0,1]$.

Define $Z^1 := (\ud \prob^1 / \ud \prob^0)|_{\F_\cdot}$; since $Z^1$ is a strictly positive $(\bF, \prob^0)$-martingale, one can write
\[
Z^1 = \exp \pare{\int_0^\cdot \lambda^1_t \ud M_t^0 - \frac{1}{2} \int_0^\cdot \inner{\lambda^1_t}{c_t \lambda^1_t} \ud G_t} N^1_t,
\]
where $N^1$ is a local $(\bF, \prob^0)$-martingale that is strongly orthogonal to $M^0$. (This multiplicative decomposition of $Z^1$ follows in a straightforward way from its corresponding additive decomposition --- see Theorem III.4.11 of \cite{MR1943877}.) It follows that $Z^\epsilon := (\ud \prob^\epsilon / \ud \prob^0)|_{\F_\cdot} = (1- \epsilon) + \epsilon Z^1$ satisfies
\[
Z^\epsilon = \exp \pare{\epsilon \int_0^\cdot \lambda^\epsilon_t \ud M_t^0 - \frac{|\epsilon|^2}{2} \int_0^\cdot \inner{\lambda^\epsilon_t}{c_t \lambda^\epsilon_t} \ud G_t} N^\epsilon_t,
\]
where $\lambda^\epsilon := (Z^1 / Z^\epsilon) \lambda^1$, and $N^\epsilon$ is a local $(\bF, \prob^0)$-martingale that is strongly orthogonal to $M^0$. With the above notation, and according to Girsanov's theorem, we have $a^\epsilon = a^0 + \epsilon \lambda^\epsilon$. 

Let $\hX^\epsilon$ denote the \num \ portfolio under market data $(\bF, \prob^\epsilon, \Real^d)$. Since there are no constraints on investment, $\hX^\epsilon$ satisfies $\hX^\epsilon_0 = 1$ and $\ud \hX^\epsilon_t = \hX^\epsilon_t \innerw{a_t^\epsilon}{\ud S_t}$ for $t \in \Real_+$; in other words, and using the Doob-Meyer decomposition of $S$ under $\prob^0$, we have
\[
\log \hX^\epsilon = \int_0^\cdot \pare{\inner{a^\epsilon_t}{ c_t a^0_t} - \frac{1}{2} \inner{a^\epsilon_t}{c_t a^\epsilon_t} }\ud G_t + \int_0^\cdot a^\epsilon_t \, \ud M^0_t.
\]
From $a^\epsilon = a^0 + \epsilon \lambda^\epsilon$ we get $\inner{a^\epsilon}{ c a^0} - (1/2) \inner{a^\epsilon}{c a^\epsilon} = \inner{a^0}{ c a^0} - (1/2) \inner{a^0}{c a^0} - (|\epsilon|^2 / 2) \inner{\lambda^\epsilon}{c \lambda^\epsilon}$; therefore,
\begin{equation} \label{eq: log-eps-der}
\frac{1}{\epsilon} \log \pare{\frac{\hX^\epsilon}{\hX^0}} = - \frac{\epsilon}{2} \int_0^\cdot \inner{\lambda_t^\epsilon}{c_t \lambda_t^\epsilon} \ud G_t + \int_0^\cdot \lambda^\epsilon_t \, \ud M^0_t.
\end{equation}
Given the above equality, and using the fact that $\ucpeps \lambda^\epsilon =  Z^1 \lambda^1 = \lambda^0$, it is straightforward that
\[
\oSeps \pare{\frac{1}{\epsilon} \log \pare{\frac{\hX^\epsilon}{\hX^0}}} = \int_0^\cdot \lambda^0_t \, \ud M^0_t.
\]

In a similar manner, one can proceed to higher-order $\epsilon$-derivatives of $\log(\hX^\epsilon)$ at $\epsilon = 0$. For example, \eqref{eq: log-eps-der} and simple algebra (remembering that $\lambda^\epsilon = \lambda^1 Z^1 / Z^\epsilon = \lambda^0 / Z^\epsilon$) gives
\begin{eqnarray*}
\frac{1}{\epsilon} \pare{\frac{1}{\epsilon} \log \pare{\frac{\hX^\epsilon}{\hX^0}} - \int_0^\cdot \lambda^0_t \, \ud M^0_t} &=& - \frac{1}{2} \int_0^\cdot \inner{\lambda_t^\epsilon}{c_t \lambda_t^\epsilon} \ud G_t + \frac{1}{\epsilon} \int_0^\cdot \pare{\lambda^\epsilon_t - \lambda^0_t }\, \ud M^0_t. \\
&=& - \frac{1}{2} \int_0^\cdot \inner{\lambda_t^\epsilon}{c_t \lambda_t^\epsilon} \ud G_t - \int_0^\cdot \lambda^\epsilon_t  \pare{Z_t^1 - 1} \, \ud M^0_t,
\end{eqnarray*}
after which it is straightforward that
\[
\oSeps \pare{\frac{1}{\epsilon} \pare{\frac{1}{\epsilon} \log \pare{\frac{\hX^\epsilon}{\hX^0}} - \int_0^\cdot \lambda^0_t \, \ud M^0_t}} = - \frac{1}{2} \int_0^\cdot \inner{\lambda_t^0}{c_t \lambda_t^0} \ud G_t - \int_0^\cdot \lambda^0_t (Z^1_t - 1) \, \ud M^0_t.
\]

\section{Proof of Theorem \ref{thm: main}} \label{sec: proof}

Throughout the proof, \emph{all} the assumptions of Theorem \eqref{thm: main} are in force. Without loss of generality, and for notational convenience, it is assumed that $\oprob = \prob^\infty$. Then, with $Z^n := (\ud \prob^n / \ud \oprob)|_{\obF}$ for all $n \in \Natura$, $\plimn Z_T^n = 1 = Z^\infty_T$ holds for all $T \in \Real_+$.

\subsection{Setting out the plan} \label{subsec: splitting}

The first step towards proving Theorem \ref{thm: main} will involve the \emph{fixed-probability} case, where $\prob^n = \oprob$ for all $n \in \Natura$. Then, the general case where \eqref{P-CONV} is assumed will be dealt with.

In order to lighten notation, we set
\begin{equation} \label{eq: hats}
a^n := a^{(\bF^n, \, \prob^n)}, \ \varphi^n := \varphi^\datan \text{ and } \hX^n := \hX^\datan, \text{ for all } n \in \Natura.
\end{equation}
For the fixed probability case, we also consider
\begin{equation} \label{eq: tildes}
\ta^n := a^{(\bF^n, \, \oprob)}, \ \tphi^n := \varphi^{(\bF^n, \, \oprob, \, \K^n)} \text{ and } \tX^n := \hX^{(\bF^n, \, \oprob, \, \K^n)}, \text{ for all } n \in \Natura.
\end{equation}
Since $\oprob = \prob^\infty$, we have $\tX^\infty = \hX^\infty$. For each $n \in \Natura$, $\tX^n$ is the \num \ portfolio for an agent with data $(\bF^n, \oprob, \K^n)$. In order to prove Theorem \ref{thm: main}, first we shall show that
\begin{equation} \label{eq: main 1}
\oSn \log \tX^n = \log \tX^\infty,
\end{equation}
and then that
\begin{equation} \label{eq: main 2}
\oSn \log (\hX^n / \tX^n) = 0.
\end{equation}

\subsection{A deterministic concave maximization problem} \label{subsec: determ conc maxim prob}
For fixed $n \in \Natura$ and $(\omega, t) \in \Omega \times \Real_+$, all $\varphi^n(\omega, t)$ defined in \eqref{eq: hats} and $\tphi^n(\omega, t)$ defined in \eqref{eq: tildes} appear as solutions to a deterministic concave maximization problem of the form
\begin{equation} \label{eq: p's}
\phi (c, \, \alpha, \, \K) \, := \, \arg \max_{f \in \K \cap \Nu^\bot} \pare{ \inner{f}{\alpha}_c - \frac{1}{2} |f|_c^2 },
\end{equation}
for some $\alpha \in \Real^d$, where the \textsl{pseudo-inner-product} $\inner{\cdot}{\cdot}_c$ on $\Real^d$ is defined via $\inner{x}{y}_c := \inner{x}{c y}$ (remember that $\inner{\cdot}{\cdot}$ is the \emph{usual} Euclidean inner-product) for all vectors $x$ and $y$ of $\Real^d$, where $c$ is a $(d \times d)$-nonnegative-definite matrix. Of course, $|\cdot|_c$ denotes the \textsl{pseudo-norm} generated by the last pseudo-inner-product $\inner{\cdot}{\cdot}_c$. In \eqref{eq: p's}, $\Nu := \{ x \in \Real^d \ | \ |x|_c = 0 \}$ and we suppose that $\Nu \subseteq \K$, so there is a unique solution to \eqref{eq: p's}, and $\phi (c, \, \alpha, \, \K)$ is well-defined.

It makes sense then to study the deterministic problem \eqref{eq: p's}. \emph{Only} for this subsection, all elements involved, including $c$ and $\K$ will be assumed deterministic. For the $(d \times d)$-nonnegative-definite matrix we shall be assuming that $\mathsf{trace}(c) = 1$, which implies in particular that $|x|_c \leq |x|$ for all $x \in \Real^d$, where $|\cdot|$ is the Euclidean norm. Observe that the $\ubF$-predictable process $(c_t)_{t \in \Real_+}$ satisfies $\mathsf{trace} (c) = 1$ since, formally, $\ud G_t = \ud \pare{\mathsf{trace} [S, S]_t} =  \mathsf{trace} \pare{ \ud [S, S]_t} = \mathsf{trace} \pare{ c_t \ud G_t} = \mathsf{trace} \pare{ c_t } \ud G_t$.

The dependence of $\phi (c, \, \alpha, \, \K)$ of \eqref{eq: p's} on $\alpha$ and $\K$ will be now examined. Remember that $B(m) := \{ x \in \Real^d \ | \ |x| \leq m \}$ for $m \in \Real_+$, as well as the definition of ``$\dist$'' from \eqref{eq: dist}.

\begin{prop} \label{prop: determin conc maxim}
Let $\alpha$, $\alpha'$ be vectors in $\Real^d$ and $\K$,  $\K'$ be closed and convex subsets of $\Real^d$ with $\Nu \subseteq \K$ and $\Nu \subseteq \K'$. With the notation of problem \eqref{eq: p's}, we have
\begin{enumerate}
  \item $|\phi (c, \, \alpha', \, \K) - \phi (c, \, \alpha, \, \K)|_c \, \leq \, |\alpha' - \alpha|_c$.
  \item $|\phi (c, \, \alpha, \, \K)|_c \, \leq \, |\alpha|_c$.
  \item $|\phi (c, \, \alpha, \, \K') - \phi (c, \, \alpha, \, \K)|^2_c \, \leq \, 4 |\alpha|_c \, \dist \left( \K' \cap B( |\alpha|_c), \, \K \cap B( |\alpha|_c) \right)$.
\end{enumerate}
\end{prop}

\begin{proof}
\noindent \textbf{(1)} Let $\phi' := \phi (c, \, \alpha', \, \K)$ and $\phi := \phi (c, \, \alpha, \, \K)$. First-order conditions for optimality imply that $\inner{\alpha' - \phi'}{\phi - \phi'}_c \leq 0$ and $\inner{\alpha - \phi}{\phi' - \phi}_c \leq 0$. Adding the previous two inequalities gives $\inner{\phi' - \phi - (\alpha' - \alpha)}{\phi' - \phi}_c \leq 0$, or, equivalently, $|\phi' - \phi|^2_c \leq \inner{\alpha' - \alpha}{\phi' - \phi}_c \leq |\alpha' - \alpha|_c |\phi' - \phi|_c$, which proves the result.

\smallskip
\noindent \textbf{(2)} Let $\phi := \phi (c, \, \alpha, \, \K)$. Since $0 \in \K$, first-order conditions give $\inner{- \phi}{\alpha - \phi}_c \leq 0$. In other words, $|\phi|_c^2 = \inner{\phi}{\phi}_c \leq \inner{\alpha}{ \phi}_c \leq |\alpha|_c |\phi|_c$, which gives $|\phi|_c \leq |\alpha|_c$.

\smallskip
\noindent \textbf{(3)} Let $\phi' := \phi (c, \, \alpha, \, \K')$ and $\phi := \phi (c, \, \alpha, \, \K)$. Let $\pr_{\K}(\phi')$ denote the projection of $\phi'$ on $\K \cap \Nu^\bot$ under the inner product $\inner{\cdot}{\cdot}_c$, and define $\pr_{\K'}(\phi)$ similarly. The projections are unique because we are restricting attention to $\Nu^\bot$. By first-order conditions, $\inner{\alpha - \phi'}{\pr_{\K'}(\phi) - \phi'}_c \leq 0$ and $\inner{\alpha - \phi}{\pr_{\K}(\phi') - \phi}_c \leq 0$, which we rewrite as $\inner{\alpha - \phi'}{\phi - \phi'}_c \leq \inner{\alpha - \phi'}{\phi - \pr_{\K'}(\phi)}_c$ and $\inner{\alpha - \phi}{\phi' - \phi}_c \leq \inner{\alpha - \phi}{\phi' - \pr_{\K}(\phi')}_c$. Adding them up and using the Cauchy-Schwartz inequality, we get $|\phi' - \phi|_c^2 \leq |\alpha - \phi'|_c |\phi - \pr_{\K'}(\phi)|_c + |\alpha - \phi|_c |\phi' - \pr_{\K}(\phi')|_c$.
Using now statement (2), the definition of ``$\dist$'' from \eqref{eq: dist}, as well as the fact that $|x|_c \leq |x|$ for all $x \in \Real^d$, statement (3) is straightforward.
\end{proof}

\begin{cor} \label{cor: conv of sols to determ conc maxim prob}
Let $(\alpha^n)_{n \in \Natura}$ be a collection of vectors of $\Real^d$ and $(\K^n)_{n \in \Natura}$ be a collection of closed, convex subsets of $\Real^d$ with $\Nu \subseteq \K^n$ for all $n \in \Natura$. If $\limn |\alpha^n - \alpha^\infty|_c = 0$ and $\climn \K^n = \K^\infty$, then $\limn |\phi (c, \, \alpha^n, \, \K^n) - \phi (c, \, \alpha^\infty, \, \K^\infty)| = 0$, in the notation of \eqref{eq: p's}.
\end{cor}

\begin{proof}
This follows directly from statements (1) and (3) of Proposition \ref{prop: determin conc maxim}, as long as one notices that $|\cdot|_c$ and $|\cdot|$ are equivalent norms on $\Nu^\bot$.
\end{proof}

\subsection{The consequence of \eqref{F-CONV}}
The purpose here is to show that the sequence $(\ta^n)_{n \in \Natural}$ of \eqref{eq: tildes} converges to $\ta^\infty$ in some sense to be made precise below. We start with Lemma \ref{lem: pred proj}, which is a result on convergence of predictable projections. Before doing so, some remarks on the extended definition of predictable projections will be given; the interested reader is referred to \cite{MR1943877} for more details. Start with some process $\chi$ that is measurable with respect to the product $\sigma$-algebra $\oF \otimes \B(\Real_+)$, where $\B(\Real_+)$ denotes the Borel-$\sigma$-algebra on $\Real_+$. Consider also some filtration $\bF = (\F_t)_{t \in \Real_+}$ satisfying \eqref{SAND}. If $\chi$ is a nonnegative process, there exists a $[0, + \infty]$-valued, $\bF$-predictable process, uniquely defined up to $\oprob$-indistinguishability, which is called the \textsl{predictable projection} of $\chi$ with respect to $(\bF, \oprob)$ and is denoted by $\chi^{\Pre(\bF, \oprob)}$, such that $\chi^{\Pre(\bF, \oprob)}_\tau  = \oexpec [\chi_\tau \such \F_{\tau -}]$ for all finite $\bF$-predictable stopping times $\tau$, where $\oexpec$ denotes expectation with respect to $\oprob$. If $\chi$ is \emph{any} $\Real$-valued measurable process, split as usual $\chi = \chi_+ - \chi_-$, where $\chi_+$ is the positive part, and $\chi_-$ the negative part, of $\chi$.  Of course, $|\chi| = \chi_+ + \chi_-$. On the $\bF$-predictable set $\{ |\chi|^{\Pre(\bF, \oprob)} < + \infty \} = \{ \chi_+^{\Pre(\bF, \oprob)} < + \infty, \, \chi_-^{\Pre(\bF, \oprob)} < + \infty \}$, define $\chi^{\Pre(\bF, \oprob)} \, := \, \chi_+^{\Pre(\bF, \oprob)} - \chi_-^{\Pre(\bF, \oprob)}$; on $\{ |\chi|^{\Pre(\bF, \oprob)} = + \infty \}$, define $\chi^{\Pre(\bF, \oprob)} \, := \, + \infty$. The extended predictable projection $\chi^{\Pre(\bF, \oprob)}$ thus defined still satisfies $\chi^{\Pre(\bF, \oprob)}_\tau = \oexpec[\chi_\tau \such \F_{\tau -}]$ for all finite $\bF$-predictable stopping times $\tau$, if one agrees that $\oexpec[\chi_\tau \such \F_{\tau -}] = + \infty$ on $\{ \oexpec[|\chi_\tau| \such \F_{\tau -}] = + \infty \}$.

\begin{lem} \label{lem: pred proj}
Consider a $\oF \otimes \B(\Real_+)$-measurable process $\chi$ such that $\int_0^T \big| \chi_t^{\Pre(\ubF, \oprob)} \big| \ud G_t < + \infty$ for all $T \in \Real_+$. If the collection $(\bF^n)_{n \in \Natura}$ satisfies \eqref{F-CONV}, we have:
\begin{enumerate}
  \item $\int_0^T \big| \chi^{\Pre(\bF^n, \oprob)} \big| \ud G_t < \infty$ for all $n \in \Natura$ and $T \in \Real_+$, and
  \item $\plimn \int_0^T \big| \chi_t^{\Pre(\bF^n, \oprob)} - \chi_t^{\Pre(\bF^\infty, \oprob)} \big| \ud G_t = 0$, for all $T \in \Real_+$.
\end{enumerate}
\end{lem}

\begin{proof}
Observe initially that, since $G$ is an increasing $\ubF$-predictable process, it suffices to show the validity of (1) and (2) for all finite $\ubF$-stopping-times $T$ such that $\oexpec[G_T] < + \infty$, instead of all deterministic times $T \in \Real_+$. Fix then a $\ubF$-stopping-time $T$ with $\oexpec[G_T] < + \infty$ and consider the positive finite measure $\mu_T$ on $\left( \Omega \times \Real_+, \, \oF \otimes \B(\Real_+) \right)$ defined via $\mu_T \left( A \, \times \, ]t_1, t_2] \right) := \oexpec \big[ \indic_A \pare{ G_{t_2 \wedge T} - G_{t_1 \wedge T}} \big]$ for $A \in \oF$ and $t_1 < t_2$ times in $\Real_+$. By a slight abuse of notation, for a measurable process $\xi$ with $\oexpec \big[ \int_0^T |\xi_t| \, \ud G_t \big] < \infty$, let $\mu_T(\xi) := \int \xi \ud \mu_T = \oexpec \big[ \int_0^T \xi_t \, \ud G_t \big]$. Note that, for any $\bF$ satisfying \eqref{SAND},
\begin{equation} \label{eq: pred proj is contr}
\mu_T \big( |\xi^{\Pre(\bF, \oprob)}| \big) \leq \mu_T \big(|\xi|^{\Pre(\bF, \oprob)}\big) = \mu_T \big(|\xi|\big), \text{ for all measurable processes } \xi.
\end{equation}
Also, it is obvious that $\limn \mu_T(|\xi^n|) = 0$ implies $\plimn \int_0^T |\xi^n_t| \, \ud G_t = 0$.

\smallskip
\noindent \textbf{(1)}
Consider the $\ubF$-predictable process $\Lambda := \int_0^\cdot |\chi_t|^{\Pre(\ubF, \oprob)} \ud G_t$. For each $m \in \Real_+$, the inequalities $\int_0^\cdot |\chi_t \indic_{\set{\Lambda_t \leq m}}|^{\Pre(\ubF, \oprob)} \ud G_t = \int_0^\cdot |\chi_t |^{\Pre(\ubF, \oprob)} \indic_{\set{\Lambda_t \leq m}} \ud G_t \leq m$ hold. Then, for $n \in \Natura$ and $m \in \Real_+$, $\mu_T \big( |\chi \indic_{\set{\Lambda \leq m}}|^{\Pre(\bF^n, \oprob)} \big) = \mu_T \big( |\chi \indic_{\set{\Lambda \leq m}}|^{\Pre(\ubF, \oprob)} \big) \leq m$. This means that, $\oprob$-a.s, $\int_0^T |\chi_t|^{\Pre(\bF^n, \oprob)} \ud G_t < \infty$ on $\set{ \Lambda_T \leq m }$ for all $n \in \Natura$ and $m \in \Real_+$. Since $\uparrow \lim_{m \to \infty} \set{\Lambda_T \leq m } = \Omega$, $\oprob$-a.s., we obtain the result of statement (1).

\smallskip
\noindent \textbf{(2)}
A process $\xi$ that is a finite linear combination of processes of the form $\indic_A \indic_{]t_1, t_2]}$ for $A \in \oF$ and $t_1 < t_2$ will be called \textsl{simple measurable}. Since $\pare{ \indic_A \indic_{]t_1, t_2]}}^{\Pre(\bF^n, \oprob)}_t = \prob[A \such \F_{t-}] \indic_{]t_1, t_2]}$ holds for all $t \in \Real_+$,
the continuity of $G$ and \eqref{F-CONV} will give
\begin{equation} \label{eq: F-CONV, simple}
\limn \mu_T \pare{ \left| \xi^{\Pre(\bF^n, \oprob)} - \xi^{\Pre(\bF^\infty, \oprob)}  \right| }= 0
\end{equation}
for any simple measurable process $\xi$. A simple density argument shows that for all measurable $\xi$ with $\mu_T(|\xi|) < \infty$ and for any $\epsilon > 0$, there exists a simple measurable $\xi'$ with $\mu_T(|\xi' - \xi|) < \epsilon$. Then, \eqref{eq: pred proj is contr} implies that \eqref{eq: F-CONV, simple} is valid whenever $\xi$ is measurable with $\mu_T(|\xi|) < \infty$. Now, pick any measurable $\chi$ that satisfies $\int_0^T \big| \chi_t^{\Pre(\ubF, \oprob)} \big| \ud G_t < + \infty$. For any $m \in \Real_+$, we have $\mu_T( |\chi \indic_{\set{\Lambda \leq m}}|) \leq m < \infty$ (remember that $\Lambda := \int_0^\cdot |\chi_t|^{\Pre(\ubF, \oprob)} \ud G_t$). Then, $\limn \mu_T \big( \big| \chi^{\Pre(\bF^n, \oprob)} \indic_{\set{\Lambda \leq m}} - \chi^{\Pre(\bF^\infty, \oprob)} \indic_{\set{\Lambda \leq m}} \big| \big) = 0$
holds by \eqref{eq: F-CONV, simple}. In other words, $\plimn \int_0^T \big| \chi^{\Pre(\bF^n, \oprob)}_t - \chi^{\Pre(\bF^\infty, \oprob)}_t \big| \ud G_t = 0$ on $\{ \Lambda_T \leq m \}$ for all $m \in \Natural$, and since $\uparrow \lim_{m \to \infty} \set{\Lambda_T \leq m } = \Omega$, $\oprob$-a.s., statement (2) is proved.
\end{proof}

\begin{cor} \label{cor: F-CONV means conv of coeffs}
We have $\plimn \int_0^T \big| c_t ( \ta^n_t - \ta^\infty_t) \big| \ud G_t = 0$ for all $T \in \Real_+$.
\end{cor}

\begin{proof}
Call $\oa := a^{(\obF, \, \oprob, \Real^d)}$. The statement of the corollary follows from Lemma \ref{lem: pred proj} as soon as one notices the following: the form of the semimartingale decomposition of $S$ under $(\bF^n, \prob^n)$, $n \in \Natural$, implies that $(c \oa)^{\Pre(\bF^n, \, \oprob)} = c \ta^n$ for all $n \in \Natura$. (Observe here that, since $c$ is predictable with respect to each $\bF^n$, $n \in \Natural$, we have $c^{\Pre(\bF^n, \, \oprob)} = c$.)
\end{proof}

\subsection{The proof of \eqref{eq: main 1}} We proceed now to show the validity of \eqref{eq: main 1}. In accordance to the deterministic notation of \S\ref{subsec: determ conc maxim prob}, for any $d$-dimensional processes $\xi$ and $\chi$ we set $\inner{\xi}{\chi}_c = (\inner{\xi_t}{\chi_t}_{c_t})_{t \in \Real_+} = (\inner{\xi_t}{c_t \chi_t} )_{t \in \Real_+}$ as well as $|\xi|_c = (|\xi_t|_{c_t})_{t \in \Real_+} = \big( \sqrt{\inner{\xi_t}{c_t \xi_t}} \big)_{t \in \Real_+}$.

For each $n \in \Natura$, write $\log \tX^n$ of \eqref{eq: tildes} in its $(\obF , \oprob)$-decomposition:
\begin{equation} \label{eq: numeraire_fixed_G}
\log \tX^n \, := \, \int_0^\cdot \pare{ \inner{\tphi^n_t}{a_t}_{c_t} - \frac{1}{2}  |\tphi^n_t|^2_{c_t} } \ud G_t + \int_0^\cdot \tphi^n_t \ud M^{(\obF, \oprob)}_t.
\end{equation}
Define $\tg^n := \inner{\tphi^n}{a}_{c} - |\tphi^n|^2_{c} / 2$ for all $n \in \Natura$; then,  $\int_0^\cdot \tg^n_t \ud G_t$ is the $(\obF, \oprob)$-growth of $\tX^n$. In order to prove \eqref{eq: main 1}, we need to show that $\plimn \int_0^T |\tg^n_t - \tg^\infty_t| \ud G_t = 0$, as well as $\plimn \int_0^T |\tphi^n_t - \tphi^\infty_t|^2_{c_t} \ud G_t = 0$.

First we show that $\plimn \int_0^T |\tg^n_t - \tg^\infty_t| \ud G_t = 0$. With some abuse of notation, let $G$ also denote the \emph{random} measure  induced by $G$ on $\Real_+$, i.e., for all $I \in \B(\Real_+)$ let $G (I) := \int_{I} \ud G_t$ . Jointly, Corollary \ref{cor: F-CONV means conv of coeffs} and Corollary \ref{cor: conv of sols to determ conc maxim prob} imply that, for all $T \in \Real_+$,
\begin{equation} \label{eq: conv of psi's}
\prob \bra{\limn |\tphi^n_t - \tphi^\infty_t|_{c_t} = 0, \text{ for } G \text{-a.e. } t \in [0, T] } = 1.
\end{equation}
This certainly implies that, for all $T \in \Real_+$, $\prob \bra{\limn \tg^n_t = \tg^\infty_t, \text{ for } G \text{-a.e. } t \in [0, T] } = 1$.
Now, if $\ophi := \varphi^{(\obF, \, \oprob, \, \Real^d)} = \oa$ and $\og := \inner{\ophi}{\oa}_c - |\ophi|^2_{c} / 2 = |\oa|^2_c / 2$, we have $0 \leq \tg^n \leq \og$ for all $n \in \Natura$, since $\int_0^\cdot \og_t \ud G_t$ is the growth of the \num \ portfolio with market data $(\obF, \oprob, \Real^d)$. The \eqref{NUPBR} condition reads $\int_0^T \og_t \ud G_t < \infty$ for all $T \in \Real_+$; therefore, in view of the dominated convergence theorem, we have $\plimn \int_0^T |\tg^n_t - \tg^\infty_t| \ud G_t = 0$.

The proof of $\plimn \int_0^T |\tphi^n_t - \tphi^\infty_t|^2_{c_t} \ud G_t = 0$ follows along the same lines. Statement (2) of Proposition \ref{prop: determin conc maxim} gives $|\tphi^n|_c \leq 2 |\oa|_c$ for all $n \in \Natura$. Since $\int_0^T |\oa_t|^2_{c_t} \ud G_t < \infty$ for all $T \in \Real_+$ from \eqref{NUPBR}, \eqref{eq: conv of psi's} gives $\plimn \int_0^T |\tphi^n_t - \tphi^\infty_t|^2_{c_t} \ud G_t = 0$, where the dominated convergence theorem was used again.

\subsection{A positive-martingale convergence result}

The next line of business is to show \eqref{eq: main 2}, and for this we have to establish that the sequence
$(a^n - \ta^n)_{n \in \Natural}$ converges to zero in some sense. For each $n \in \Natural$, define the density process $Z^n := (\ud \prob^n / \ud \oprob) |_{\bF^n}$, and consider the following multiplicative decomposition of $Z^n$, following from its corresponding additive decomposition, as is presented for example in Theorem III.4.11 (page 182) of \cite{MR1943877}:
\begin{equation} \label{eq: R-N densities}
Z^n = \exp \pare{ \int_0^\cdot \zeta^n_t \, \ud M_t^{(\bF^n, \oprob)} - \frac{1}{2} \int_0^\cdot |\zeta^n_t|^2_{c_t} \, \ud G_t} N^n.
\end{equation}
Here, for each $n \in \Natural$, $N^n$ is a strictly positive $\bF^n$-local $\oprob$-martingale with $[M^{(\bF^n, \oprob)}, N^n] = 0$, i.e., $N^n$ is strongly orthogonal to $M^{(\bF^n, \oprob)}$.
A simple application of Girsanov's theorem shows that $\zeta^n = a^n - \ta^n$. Therefore, we first have to establish some result that connects convergence of $(Z^n)_{n \in \Natural}$ to $Z^\infty = 1$ to convergence to zero of the quadratic variation of their stochastic logarithms. This is done in Theorem \ref{thm: conv of pos loc marts to one}. Then, Corollary \ref{cor: convergence of zeta's to zero} gives us convergence to zero of $(\zeta^n)_{n \in \Natural}$, in an appropriate sense.

In the course of the proof of Theorem \ref{thm: conv of pos loc marts to one} we make use of (one side of) the \textsl{Davis inequality}. Namely, if $L$ is a one-dimensional $\bF$-local $\oprob$-martingale with quadratic variation $[L, L]$, then $\oexpec \big[ \sqrt{[L, L]_T} \big] \leq 6 \oexpec \big[ \sup_{t \in [0, T]} |L_t| \big]$ for all $T \in \Real_+$; see Theorem 4.2.12, page 213 of \cite{MR1906715}. (Remember that $\oexpec$ denotes expectation under the probability $\oprob$.) In particular, if a sequence of $(L^n)_{n \in \Natural}$, where each $L^n$ is a $\bF^n$-local $\oprob$-martingale for each $n \in \Natural$, satisfies $\Lb^1(\oprob)$-$\limn \sup_{t \in [0, T]} |L^n_t| = 0$, then also $\Lb^1(\oprob)$-$\limn \sqrt{[L^n, L^n]_T} = 0$.
\begin{thm} \label{thm: conv of pos loc marts to one}
Consider a sequence $(Z^n)_{n \in \Natural}$ of \cadlag \ processes, such that:
\begin{itemize}
  \item $Z^n_0 = 1$ and $Z^n_t > 0$ for all $t \in \Real_+$, $\oprob$-a.s., for all $n \in \Natural$.
  \item Each $Z^n$ is a $\bF^n$-local $\oprob$-martingale.
  \item $\plimn Z^n_T = 1$ for all $T \in \Real_+$.
\end{itemize}
Then, we have the following:
\begin{enumerate}
  \item $\Lb^1(\oprob)$-$\limn Z_T^n = 1$ for all $T \in \Real_+$.
  \item $\ucpn Z^n = 1$.
  \item $\plimn [Z^n, Z^n]_T = 0$ for all $T \in \Real_+$.
  \item $\plimn [R^n, R^n]_T = 0$ for all $T \in \Real_+$, where $R^n \, := \, \int_0^\cdot (1 / Z^n_{t-}) \ud Z^n_t$, i.e., $R^n$ is the \textsl{stochastic logarithm} of $Z^n$, for $n \in \Natural$.
\end{enumerate}
\end{thm}

\begin{proof}
\noindent \textbf{(1)} Since $\oexpec[Z^n_T] \leq 1$ for all $n \in \Natural$, it is a consequence of Fatou's lemma that $\limn \oexpec [Z_T^n] = 1$ for all $T \in \Real_+$. Theorem 16.14(ii), page 217 in \cite{MR1155402} implies the $\oprob$-uniform integrability of $(Z^n_T)_{n \in \Natural}$. We thus obtain $\Lb^1(\oprob)$-$\limn Z_T^n = 1$ for all $T \in \Real_+$.

\smallskip
\noindent \textbf{(2)} Fix $T \in \Real_+$. We first show that $\plimn \sup_{t \in [0, T]} Z_t^n = 1$; in the next paragraph we will establish that $\plimn \inf_{t \in [0, T]} Z_t^n = 1$, which completes the proof of the statement. Fix $\epsilon > 0$ and $T \in \Real_+$ and define the $\bF^n$-stopping-time $\tau^n :=  \inf \{t \in [0, T] \such Z^n_t > 1 + \epsilon \} \wedge T$ for all $n \in \Natural$. Since $\oexpec[Z^n_T] \leq \oexpec [Z^n_{\tau^n}] \leq 1$ by the optional sampling theorem (see for example \S1.3.C of \cite{MR1121940}), it follows that $\limn \oexpec [Z^n_{\tau^n}] = 1$. Showing that $\limn \oprob[\tau^n < T] = 0$ will imply that $\plimn \sup_{t \in [0, T]} Z^n_t = 1$, since $\epsilon > 0$ is arbitrary. Suppose on the contrary (passing to a subsequence if necessary) that $\limn \oprob[\tau^n < T] = \delta > 0$.  Then, since $\left| \oexpec \bra{ Z^n_{T} \indic_{\{ \tau^n = T \}} } - \oprob[\tau^n = T] \right| =  \left| \oexpec \bra{(Z_T^n - 1) \indic_{\{ \tau^n = T \}} } \right| \leq \oexpec [|Z_T^n - 1|]$,
and the last quantity converges to zero as $n \to \infty$, we get $\limn \oexpec \bra{ Z^n_{T} \indic_{\{ \tau^n = T \}} } = 1 - \delta$. In turn, this implies
\begin{eqnarray*}
  1 = \limn \oexpec [Z^n_{\tau^n}] &\geq& \liminf_{n \to \infty} \oexpec \big[ Z^n_{\tau^n} \indic_{\{ \tau^n < T \}} \big] + \limn \oexpec \big[ Z^n_{T} \indic_{\{ \tau^n = T \}} \big] \\
    &\geq& (1 + \epsilon) \delta + (1 - \delta) = 1 + \epsilon \delta,
\end{eqnarray*}
which contradicts the fact that $\delta > 0$. Thus, $\plimn \sup_{t \in [0, T]} Z_t^n = 1$ has been shown.

Now, to prove $\plimn \inf_{t \in [0, T]} Z_t^n = 1$ for fixed $T \in \Real_+$. Fix some $\epsilon > 0$, and for each $n \in \Natural$, redefine $\tau^n :=  \inf \{t \in [0, T] \such Z^n_t < 1 - \epsilon \} \wedge T$ --- we only need to show that $\limn \oprob[\tau^n < T] = 0$. Observe that on the event $\{ \tau^n < T \}$ we have $\oprob[Z^n_T > 1 - \epsilon^2 \such \F_{\tau^n}] \leq (1 - \epsilon) / (1 - \epsilon^2) = 1 / (1 + \epsilon)$. Then,
\[
\oprob[Z^n_T > 1 - \epsilon^2] = \oexpec \big[ \oprob[Z^n_T > 1 - \epsilon^2 \such \F_{\tau^n}] \big] \, \leq \, \oprob[\tau^n = T] + \oprob[\tau^n < T] \frac{1}{1+\epsilon}.
\]
Using $\oprob[\tau^n = T] = 1- \oprob[\tau^n < T]$, rearranging the previous inequality and taking the superior limit as $n \to \infty$, we get
\[
\limsup_{n \to \infty} \oprob[\tau^n < T] \ \leq \ \frac{1 + \epsilon}{\epsilon} \limsup_{n \to \infty} \oprob[Z_T^n \leq 1 - \epsilon^2] \ = \ 0,
\]
which completes the proof of statement (2).

\smallskip
\noindent \textbf{(3)} Fix some $T \in \Real_+$ and let $\tau^n := \inf \{ t \in \Real_+ \such Z_t^n > 2 \} \wedge T$; each $\tau^n$ is a $\bF^n$-stopping time. Let $Y^n$ be defined via $Y^n_t = Z^n_{t \wedge \tau^n} - \Delta Z^n_{\tau^n} \indic_{\{ \tau^n \leq t \}}$; in other words, $Y^n$ is the process $Z^n$ stopped \emph{just before} time $\tau^n$. Since $\Delta Z^n_{\tau^n} \geq 0$, $Y^n$ is a $(\oprob, \bF^n)$-supermartingale and $0 \leq Y^n \leq 2$ holds for all $n \in \Natural$. Since $\limn \oprob [ \tau^n = T ] = 1$, as well as $\plimn \Delta Z^n_{\tau^n} = 0$ holding in view of statement (2), for statement (3) to hold it suffices to show that $\plimn [Y^n, Y^n]_T = 0$. For each $n \in \Natural$, write $Y^n = - B^n + L^n$ for the Doob-Meyer decomposition of $Y^n$ under $(\bF^n, \oprob)$. Since, for each $n \in \Natural$, $Y^n$ is a uniformly bounded $(\bF^n, \oprob)$-supermartingale, $B^n$ is increasing, $\oprob$-integrable and $\bF^n$-predictable, while $L^n$ is a $(\bF^n, \oprob)$-martingale with $L^n_0=1$. Now, $L^n_0=1$, $L^n \geq Y^n$ and $\plimn Y^n_T = 1$ imply that $\plimn L^n_T = 1$; otherwise $\limsup_{n \to \infty} \oexpec [L^n_T] > 1$, which is impossible. Using $\plimn Y^n_T = 1$ and $\plimn L^n_T = 1$, we get $\plimn B^n_T = 0$. Note that both sequences $(Y^n_T)_{n \in \Natural}$ and $(L^n_T)_{n \in \Natural}$ are $\oprob$-uniformly integrable; the first because it is uniformly bounded; the second because it is actually converging in $\Lb^1(\oprob)$ according to statement (1) of this Theorem. This means that $(B^n_T)_{n \in \Natural} = (L^n_T - Y^n_T)_{n \in \Natural}$ is $\oprob$-uniformly integrable as well. Since $\sup_{t \in [0, T]} |L^n_t - 1 | \leq \sup_{t \in [0, T]} |Y^n_t - 1 | + B^n_T \leq 1  + B^n_T$, this further means that the collection $\big( \sup_{t \in [0, T]} |L^n_t - 1 | \big)_{n \in \Natural}$ is  $\prob$-uniformly integrable as well. As, by statement (2) of this Theorem, $\ucpn L^n = 1$, we actually have $\Lb^1(\oprob)$-$\limn \sup_{t \in [0, T]} |L^n_t - 1 | = 0$. The Davis inequality now gives $\Lb^1(\oprob)$-$\limn \sqrt{[L^n, L^n]_T} = 0$, which implies $\plimn [L^n, L^n]_T = 0$. Finally, since
\[
[B^n, B^n]_T - 2 [L^n, B^n]_T = - [B^n + 2 Y^n, B^n]_T \leq - 2[Y^n, B^n]_T = - 2 \sum_{t \in ]0, T]} \Delta Y^n_t \Delta B^n_t \leq 4 B^n_T,
\]
the last inequality holding because $\Delta Y^n \geq - 2$, we are able to estimate $[Y^n, Y^n]_T = [L^n, L^n]_T + [B^n, B^n]_T - 2 [L^n, B^n]_T \leq [L^n, L^n]_T + 4 B_T^n$. Therefore, $\plimn [Y^n, Y^n]_T = 0$, which finishes the proof of statement (3).

\smallskip
\noindent \textbf{(4)} Given statements (2) and (3), (4) readily follows since $[Z^n, Z^n] = \int_0^\cdot |Z^n_t|^2 \ud [R^n, R^n]_t$.
\end{proof}

\begin{rem}
Theorem \ref{thm: conv of pos loc marts to one} is valid under the weaker assumptions:
\begin{itemize}
  \item $Z^n_0 = 1$ and $Z^n_t \geq 0$ for all $t \in \Real_+$, $\oprob$-a.s., for all $n \in \Natural$.
  \item Each $Z^n$ is a $(\bF^n, \, \oprob)$-supermartingale.
  \item $\plimn Z^n_T = 1$ for all $T \in \Real_+$.
\end{itemize}
However, we have to make some sense of the stochastic logarithms $R^n$ in the case where $Z^n$ might become zero. For each $n \in \Natural$ and $\epsilon > 0$, define the $\bF^n$-stopping-time $\kt^n (\epsilon) \, := \, \inf \set{t \in \Real_+ \ | \ Z^n_t \leq \epsilon}$. There exists a $\bF^n$-local $\prob$-supermartingale $R^n(\epsilon)$ with $R_0^n(\epsilon) = 0$ such that $\ud Z^n_t = Z^n_t \ud R^n_t(\epsilon)$ for $t \in [0, \kt^n(\epsilon)]$. It is straightforward to see that for $\epsilon' < \epsilon$ we have $\kt^n(\epsilon) \leq \kt^n(\epsilon')$ and that $R^n_t(\epsilon) = R^n_t(\epsilon')$ for for $t \in [0, \kt^n(\epsilon)]$. We can then define a process $R^n$ on the stochastic interval $\Gamma^n \, := \, \bigcup_{\epsilon > 0} [0, \kt^n(\epsilon)]$ such that $\ud Z^n_t = Z^n_t \ud R^n_t$ for all $t \in \Gamma^n$; we call this $R^n$ the \textsl{extended stochastic logarithm} of $Z^n$. Since $\plimn Z^n_T = 1$ for all $T \in \Real_+$, we get that $\plimn \sup (\Gamma^n) = + \infty$; therefore, there is no problem in the pathwise definition of $R^n$ for compact intervals of $\Real_+$ as $n \to \infty$. In this sense, statement (4) of Theorem \ref{thm: conv of pos loc marts to one} follows.
\end{rem}

\begin{cor} \label{cor: convergence of zeta's to zero}
In the notation of \eqref{eq: hats} and \eqref{eq: tildes}, $\plimn \int_0^T |a^n_t - \tilde{a}^n_t|^2_{c_t} \, \ud G_t = 0$ holds for all $T \in \Real_+$.
\end{cor}

\begin{proof}
For all $n \in \Natural$, $Z^n$ as defined in \eqref{eq: R-N densities} is a $(\bF^n, \, \oprob)$-martingale. \eqref{P-CONV} implies that $\plimn Z^n_T = 1$ holds for all $T \in \Real_+$. In the notation of Theorem \ref{thm: conv of pos loc marts to one}, we have $\int_0^\cdot |\zeta^n_t|^2_{c_t} \, \ud G_t \leq [R^n, R^n]$. The result follows because $\zeta^n = a^n - \ta^n$ for all $n \in \Natural$, and $\plimn [R^n, R^n]_T = 0$ holds for all $T \in \Real_+$ by Theorem \ref{thm: conv of pos loc marts to one}.
\end{proof}

\subsection{The proof of \eqref{eq: main 2}} We now finish the proof of Theorem \ref{thm: main} by showing \eqref{eq: main 2}. The semimartingale decomposition of $\log (\hX^n / \tX^n)$ under $(\obF, \oprob)$ reads
\[
\log \pare{ \frac{\hX^n}{\tX^n} } \, = \, \int_0^\cdot \pare{ \inner{\varphi_t^n - \tphi_t^n}{\oa_t}_{c_t} - \frac{1}{2} \pare{ |\varphi^n_t|^2_{c_t} - |\tphi^n_t|^2_{c_t} } } \ud G_t + \int_0^\cdot (\varphi_t^n - \tphi_t^n) \ud M_t^{(\obF, \oprob)}
\]
Since $a^n = \ta^n + \zeta^n$, statement (1) of Proposition \ref{prop: determin conc maxim} implies that $|\varphi^n - \tphi^n|_c \leq |\zeta^n|_c = |a^n - \ta^n|_c$. The quadratic variation of $\int_0^\cdot (\varphi_t^n - \tphi_t^n) \ud M_t^{(\obF, \oprob)}$ is equal to $\int_0^\cdot |\varphi_t^n - \tphi_t^n|^2_{c_t} \, \ud G_t \leq \int_0^\cdot |a_t^n - \ta_t^n|^2_{c_t} \, \ud G_t$. Therefore, Corollary \ref{cor: convergence of zeta's to zero} gives that
\[
\plimn \bra{ \int_0^\cdot (\varphi_t^n - \tphi_t^n) \ud M_t^{(\obF, \oprob)}, \, \int_0^\cdot (\varphi_t^n - \tphi_t^n) \ud M_t^{(\obF, \oprob)}}_T  = 0, \text{ for all } T \in \Real_+.
\]
Furthermore, for fixed $T \in \Real_+$,
\[
\oprob \bra{\limn \pare{\inner{\varphi_t^n - \tphi_t^n}{\oa_t}_{c_t} - \frac{1}{2} \pare{ |\varphi^n_t|^2_{c_t} - |\tphi^n_t|^2_{c_t} } } = 0, \text{ for } G \text{-a.e. } t \in [0, T] } = 1.
\]
One can then use the domination relationship $\left| \inner{\varphi^n - \tphi^n}{\oa}_{c} -  \pare{ |\varphi^n|^2_{c} - |\tphi^n|^2_{c} } / 2 \right| \, \leq \, 2 \og$
to actually get that $\plimn \int_0^T \left| \inner{\varphi_t^n - \tphi_t^n}{\oa_t}_{c_t} - \pare{ |\varphi^n_t|^2_{c_t} - |\tphi^n_t|^2_{c_t} } / 2\right| \, \ud G_t = 0$, for all $T \in \Real_+$, and finish the proof of \eqref{eq: main 2}.

\bibliographystyle{siam}
\bibliography{num_stab}
\end{document}